\documentclass[conference]{IEEEtran}

\usepackage{amssymb}
\usepackage{algorithmic}
\usepackage{algorithm}
\usepackage{verbatim} 
\usepackage[mathcal]{euscript}
\usepackage{mathrsfs}
\usepackage{cite}

\newtheorem{theorem}{Theorem}
\newtheorem{proposition}{Proposition}

\ifCLASSINFOpdf
\else

   \usepackage[dvips]{graphicx}
   \graphicspath{{figures/}}
   \DeclareGraphicsExtensions{.eps}
\fi

\usepackage{amsmath}\allowdisplaybreaks
\usepackage{amsmath}
\usepackage[caption=false,font=footnotesize]{subfig}

\allowdisplaybreaks[0]

\begin{document}

\title{Half-Duplex Relaying for the Multiuser Channel}

\author{Ming Lei and Mohammad Reza Soleymani\\
Electrical and Computer Engineering\\Concordia University, Montreal, Quebec, Canada\\
Email:m\_lei,msoleyma@ece.concordia.ca
}

\maketitle

\begin{abstract}
This work focuses on studying the half-duplex (HD) relaying in the Multiple Access Relay Channel (MARC) and the Compound Multiple Access Channel with a Relay (cMACr). A generalized Quantize-and-Forward (GQF) has been proposed to establish the achievable rate regions. Such scheme is developed based on the variation of the Quantize-and-Forward (QF) scheme and single block with two slots coding structure. The results in this paper can also be considered as a significant extension of the achievable rate region of Half-Duplex Relay Channel (HDRC). Furthermore, the rate regions based on GQF scheme is extended to the Gaussian channel case. The scheme performance is shown through some numerical examples.

\end{abstract}

\section{Introduction}

Relaying has been shown to be beneficial to a conventional point-to-point communication channel by cooperating with the transmitter \cite{Cover1979}. Meanwhile, it is also proved that the relaying can improve the sum achievable rates for a multiple access channel \cite{Kramer2005} and a compound multiple access channel \cite{Gunduz2010}. 

One of the fundamental relaying schemes proposed in \cite{Cover1979} is compress-and-forward (CF). Comparing to the decode-and-forward (DF) based schemes, CF based schemes are not limitated by the decoding capability of the relay as studied in \cite{Cover1979} and \cite{Gunduz2010}. Different variations of the CF based schemes have been investigated in \cite{Cover2007,Razaghi2013,Avestimehr2011,Lim2011}.
In \cite{Lim2011}, a Noisy network coding (NNC) scheme was proposed. NNC is able to  recover the rate achieved by the classsic CF in a three-node relay channel and generally outperforms other CF based schemes in a multiuser channel. On the other hand, NNC requires the "long" sources messages to be repeated several times which significantly increases the decoding delay. In order to overcome this drawback, a short message noisy network coding (SNNC) was introduced by Wu and Xie\cite{Wu2010}\cite{Wu2013}. In SNNC, each source transmits independent short messsage in each block. In \cite{Hou2012} and \cite{Kramer2011}, SNNC was also shown to achieve the same rate as NNC in multiple multicast sessions.


Usually, the SNNC is applied in the full-duplex relay channel. Motivated by the practical constraint that relay cannot transmit and receive simultaneously in wireless communications \cite{Laneman2003}, a quantize-and-forward (QF) scheme, originates from NNC \cite{Lim2011}, has been studied for a fading half-duplex relay channel (HDRC) in \cite{Yao2013}. A single block and two slots coding based QF scheme for the HDRC has been proposed in \cite{Yao2013} to derive the achievable rates. In this paper, the "cheap" half-duplex relay \cite{Khojastepour2003} is also considered. Specifically, the achievable rate regions for the half-duplex MARC and cMACr, as shown in Fig.\ref{fig:MARC-phases} and Fig.\ref{fig:phases}, respectively, are investigated based on a variation of the QF scheme, namely the GQF scheme. Compared with the QF scheme, the proposed GQF scheme not only adopts the single block two slots coding structure but also takes into account the effect of the co-exist interfered message signal at relay. Comparing with the classic CF based schemes, the benefit is that the proposed GQF scheme is able to simplify the operation at relay ("cheaper" relay) while keeping the advantage of the CF based schemes. The GQF scheme only requires a simplified relay in the sense that no Wyner-Ziv binning is necessary. Correspondingly, the GQF scheme performs joint decoding instead of sequential decoding. The GQF scheme simplifies the relay encoding and can be implemented in any situations where a low-cost half-duplex relay is needed.

\begin{figure}[htpb]
\centering
\subfloat[MARC Slot-1]           {\label{fig:MARC-phase1}
\includegraphics[width=1.5in]{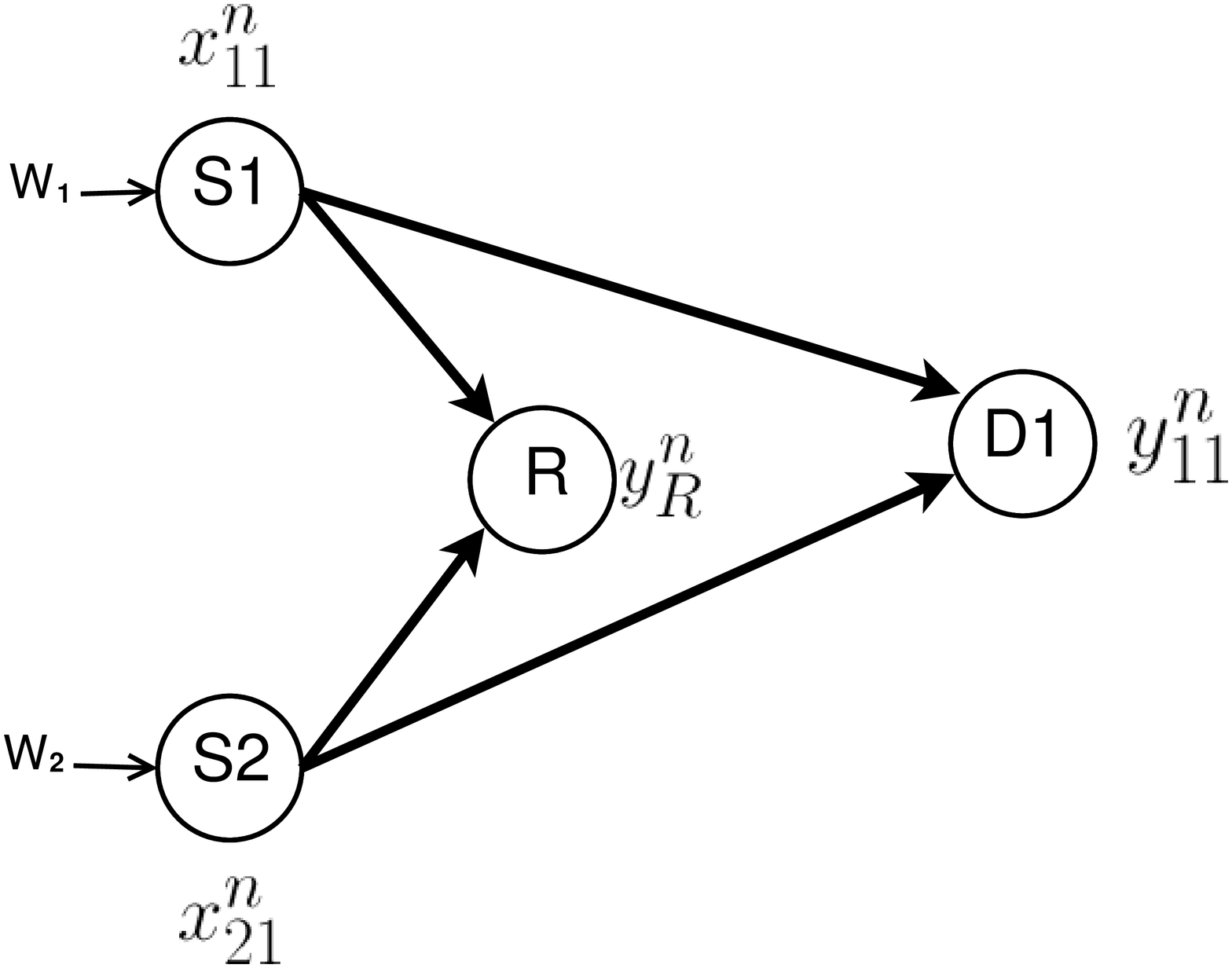}   }
\qquad
\subfloat[MARC Slot-2]{\label{fig:MARC-phase2}
\includegraphics[width=1.5in]{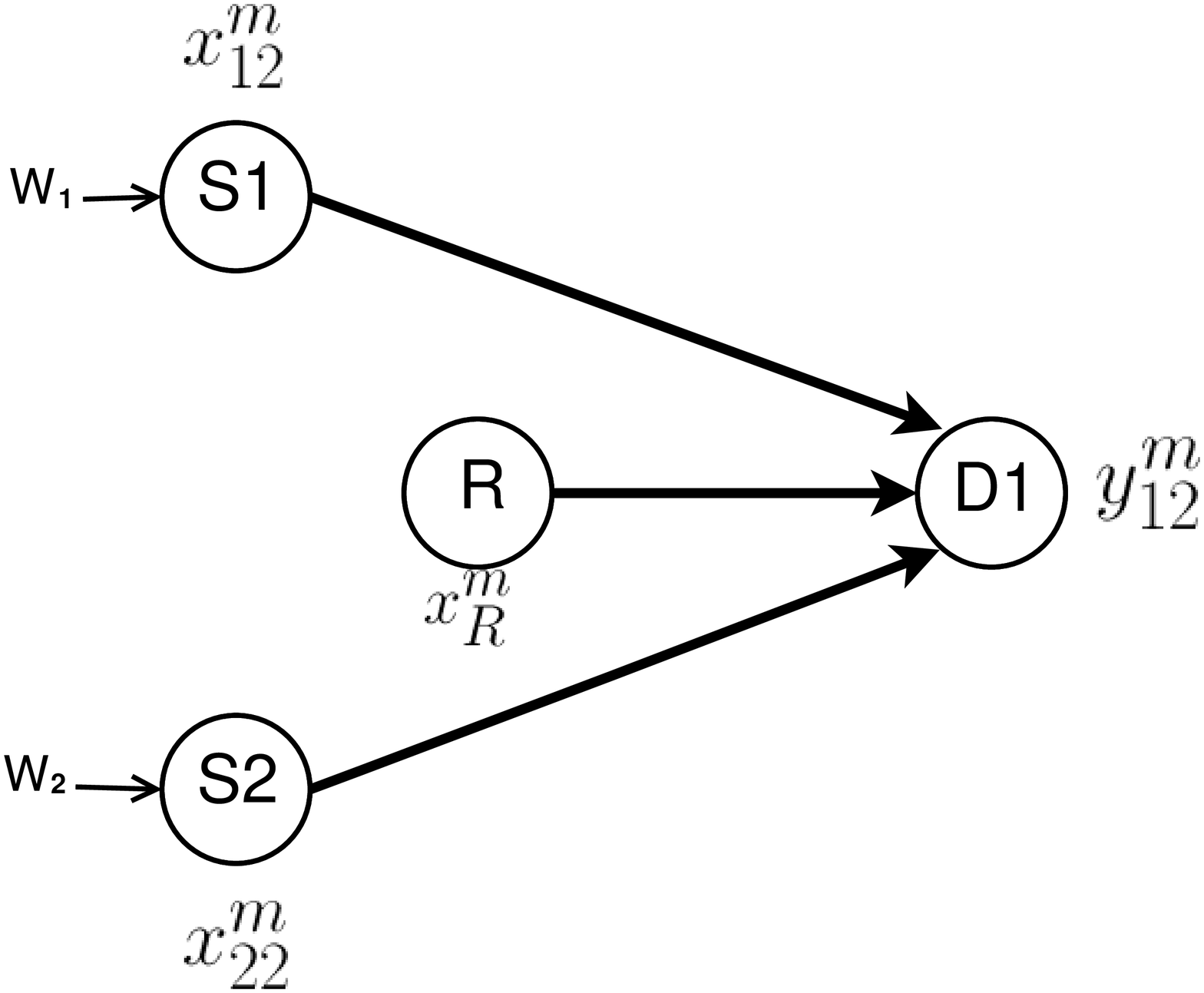}   }
\caption{Message flow of the Half-Duplex Multiple Access Relay Channel. $R$ listens to the channel in Slot-1 and Transmits to $D_1$ in Slot-2}
\label{fig:MARC-phases}
\end{figure}

\begin{figure}[htpb]
\centering
\subfloat[Slot-1]           {\label{fig:phase1} \includegraphics[width=1.5in]{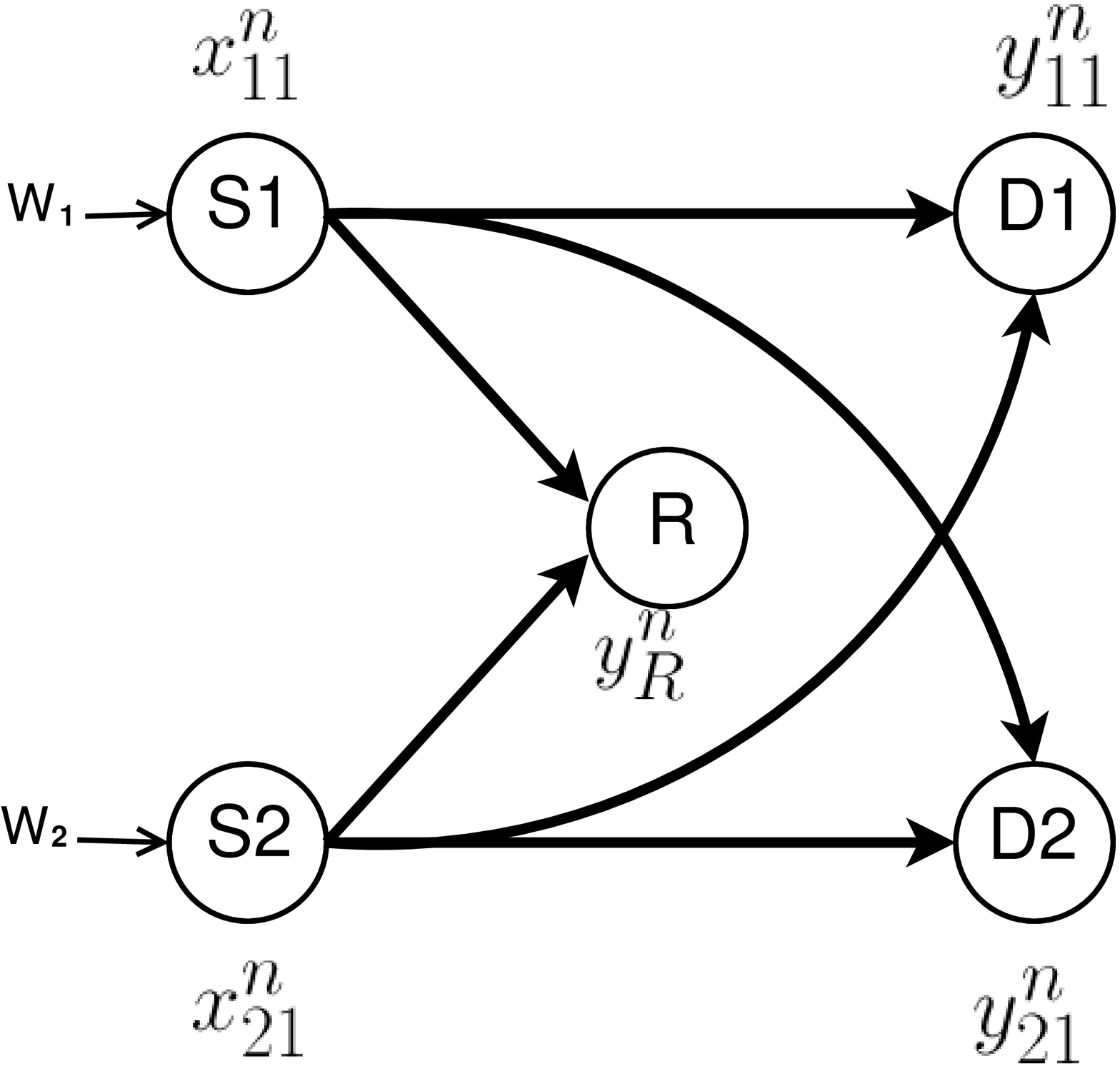}   }
\qquad
\subfloat[Slot-2]{\label{fig:phase2}
\includegraphics[width=1.5in]{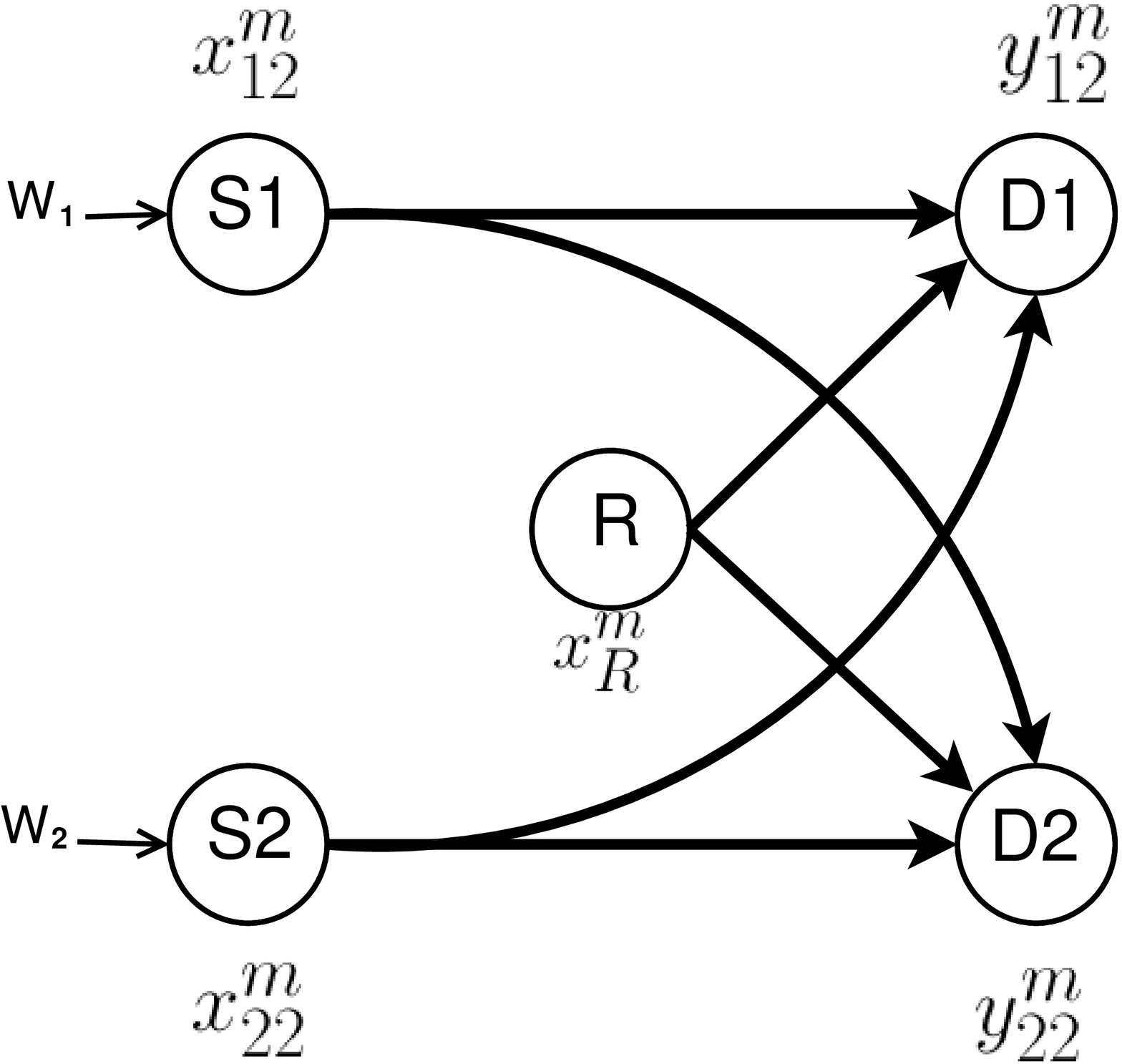}   }
\caption{Message flow of the cMACr. $R$ listens to the channel in Slot-1 and Transmits to both $D_1$ and $D_2$ in Slot-2}
\label{fig:phases}
\end{figure}

For the channels as shown in Fig.\ref{fig:MARC-phases} and Fig.\ref{fig:phases} using GQF, the relay quantizes the signal received from both sources and forwards the quantization index without using Wyner-Ziv coding (binning). Each destination decodes both messages from checking the joint typicality of the signal received in both slots without decoding the quantization index.

The performance comparison between the GQF scheme and the CF scheme is also discussed in this paper. However, in order to fit in the half-duplex MARC and cMACr, the classic CF scheme has been modified. Furthermore, the rate regions based on GQF scheme are extended from the discrete memoryless channel to the Gaussian channel. The scheme performance is shown through some numerical examples.

\section{System Model}

Two discrete memoryless channel models, a half-duplex multiple access relay channel and a half-duplex compound multiple access channel with a relay, are considered in this work. Since a cMACr naturally reduces to a MARC if the second destination is not present, in the following only the description for the HD-cMACr will be shown. The HD-MARC can be described by adjusting the random variables based on the absence of the second destination.

As shown in Fig. \ref{fig:phases}, a cMACr consists two sources $S_{1}$, $S_{2}$ and two destinations $D_{1}$ and $D_{2}$. Relay $R$ helps the information propagation from sources to destinations by cooperating with the sources. Relay operates in the Half-Duplex mode. This means that $R$ is either receiving signals from the source nodes ($S_{1}$ and $S_{2}$) or transmitting to the destinations ($D_{1}$ and $D_{2}$). Assume that each block length is totally $l$ channel uses and has two slots. The first slot ($R$ listens to the channel) and the second slot ($R$ transmits to the channel) of a single block are of $n$ and $m$ channel uses, respectively. Hence, the number of channel uses in one block is $l=n+m$.

Each source $S_i$, $i=1,2$ chooses a message $W_i$ from a message set $\mathcal{W}_{i}=\{1,2,\dots,2^{lR_i}\}$, then encodes this message into a length $n$ codeword with an encoding function $f_{i1}(W_i)=X_{i1}^{n}$ and a length $m$ codeword with an encoding function $f_{i2}(W_i)=X_{i2}^{m}$, finally sends these two codewords in the corresponding slots. Relay $R$ employs an encoding function based on its recetpion $Y_{R}^{n}$ in the first slot.
Each destination uses a decoding function $g_{i}(Y_{i1}^{n},Y_{i2}^{m})=(\hat{W}_1,\hat{W}_2)$ that jointly decodes messages from the receptions in both slots. 
The channel is memoryless such that the channel transition probabilities can be represented by
\begin{equation}
\begin{split}
& p_{Y_{R}^{n}Y_{11}^{n}Y_{21}^{n}| X_{11}^{n}X_{21}^{n}}(y_R^n,y_{11}^n,y_{21}^n| x_{11}^n,x_{21}^n)\\
&=\prod_{i=1}^{n}p_{Y_{R}Y_{11}Y_{21}| X_{11}X_{21}}(y_{R,i},y_{11,i},y_{21,i}| x_{11,i}x_{21,i})\\
\end{split}
\end{equation}
and
\begin{equation}
\begin{split}
& p_{Y_{12}^{m}Y_{22}^{m}| X_{12}^{m}X_{22}^{m}X_{R}^{m}}(y_{12}^m,y_{22}^m| x_{12}^m,x_{22}^m,x_{R}^m)\\
& =\prod_{i=1}^{m}p_{Y_{12}Y_{22}| X_{12}X_{22}X_{R}}(y_{12,i},y_{22,i}| x_{12,i}x_{22,i}x_{R,i}).\\
\end{split}
\end{equation}

A rate pair $(R_1,R_2)$ is called achievable if there exists a message set, together with the encoding and decoding functions stated before such that $Pr(\hat{W}_1\neq W_1 \cup \hat{W}_2\neq W_2)\rightarrow 0$ when $l \rightarrow \infty $.

\section{Main Result}
In this section, the achievable rate regions based on the GQF scheme for the discrete memoryless half-duplex MARC and cMACr are presented first. As a reference, the achievable rates based on a modified CF scheme is shown in the second subsection. In the last part of this section, it is shown that the achievable rate regions for the three node HDRC \cite{Yao2013} can be treated as a special case of the result for our five node cMACr.

\subsection{Achievable Rate Region Based on GQF Scheme}
In this subsection, the achievable rate regions for the discrete memoryless HD-MARC and HD-cMACr will be shown based on the GQF Scheme. The GQF scheme is  an essential variation of the classic CF. In GQF scheme, relay quantizes its observation after the first slot, and then sends the quantization index in the second slot. Unlike the conventional CF, no Wyner-Ziv binning is applied in the relay, which simiplifies the relay operation. At the destination, decoding is also different in the sense that joint-decoding of the messages from both slots without explicitly decoding the quantization index is performed in GQF scheme.

\subsubsection{Achievable Rate Region for discrete memoryless HD-MARC}
In the HD-MARC, there is only one destination $D_1$ comparing to HD-cMACr. $D_1$ receives signals from $S_1$ and $S_2$ in the first slot. It then receives from $S_1$, $S_2$ and $R$ in the second slot. The decoding is done by $D_1$ found both messages $W_1$ and $W_2$ that were sent from the sources. The following theorem describes the achievable rate region for this discrete memoryless HD-MARC:

\begin{theorem}\label{th-QFD-MARC}
The following rate regions are achievable over discrete memoryless HD-MARC based on the GQF scheme:
{\setlength\arraycolsep{0.1em}
\begin{eqnarray}
R_i & < & min\{a_1(i),b_1(i)\}
\label{eqn-QFD111-MARC} \\
R_1+R_2 & < & min \{c_1,d_1\}
\label{eqn-QFD222-MARC}
\end{eqnarray}
}
where  $\beta=n/l$ is fixed,
{\setlength\arraycolsep{0.1em}
\begin{eqnarray}
&a_k(i)=&\beta I(X_{i1};X_{j1},Y_{k1},\hat{Y}_R) + (1-\beta)I(X_{i2};X_{j2},X_R,Y_{k2})
\nonumber \\
&b_k(i)=&\beta[I(X_{i1};X_{j1},Y_{k1})-I(\hat{Y}_R;Y_R | X_{i1},X_{j1},Y_{k1})] 
\nonumber \\
&& \quad + (1-\beta)I(X_{i2},X_{R};X_{j2},Y_{k2})
\nonumber\\
&c_k=&\beta I(X_{11},X_{21};Y_{k1},\hat{Y}_R) + (1-\beta)I(X_{12},X_{22};X_R,Y_{k2})
\nonumber \\
&d_k=&\beta[I(X_{11},X_{21},\hat{Y}_{R};Y_{k1})+I(X_{11},X_{21};\hat{Y}_R)\nonumber \\
&&-I(Y_R;\hat{Y}_R)]+(1-\beta)I(X_{12},X_{22},X_{R};Y_{k2}),\label{set-abcd-MARC}
\end{eqnarray}
}
$i,j,k\in \{1,2\}$ and $i\neq j$,
for all input distributions
\begin{equation}
p(x_{11})p(x_{21})p(x_{12})p(x_{22})p(x_R)p(\hat{y}_R|y_R).
\label{inputD-MARC}
\end{equation}
\end{theorem}
\begin{proof}: As stated in section \MakeUppercase{\romannumeral 2}, HD-MARC can be treated as a reduced case of HD-cMACr. Thus, the above results can be obtained by simplifying the proof for the theorem \ref{th-QFD} under the assumption that the receiving random variables $Y_{21}=Y_{22}=\phi$ at $D_2$. The detail of the proof for theorem \ref{th-QFD} can be found in the first part of the Appendix.
\end{proof}
\subsubsection{Achievable Rate Region for discrete memoryless HD-cMACr}
In the HD-cMACr, each destination $D_i$, $i=1,2$, tries to decode both messages $W_1$ and $W_2$. The decoding is done by each of the destination found both messages. 
The following theorem describes the achievable rate region for this discrete memoryless HD-cMACr:
\begin{theorem}\label{th-QFD}
The following rate regions are achievable over discrete memoryless half-duplex cMACr with the GQF scheme:
{\setlength\arraycolsep{0.1em}
\begin{eqnarray}
R_i & < & min\{a(i),b(i)\}
\label{eqn-QFD111} \\
R_1+R_2 & < & min \{c,d\}
\label{eqn-QFD222}
\end{eqnarray}
}
where  $\beta=n/l$ is fixed, $a(i)=min \{a_{1}(i),a_{2}(i)\}$, $b(i)=min \{b_{1}(i),b_{2}(i)\}$, $c=min\{c_1,c_2\}$, $d=min\{d_1,d_2\}$,
{\setlength\arraycolsep{0.1em}
\begin{eqnarray}
a_k(i)&=&\beta I(X_{i1};X_{j1},Y_{k1},\hat{Y}_R) + (1-\beta)I(X_{i2};X_{j2},X_R,Y_{k2})
\nonumber \\
b_k(i)&=&\beta[I(X_{i1};X_{j1},Y_{k1})-I(\hat{Y}_R;Y_R | X_{i1},X_{j1},Y_{k1})] 
\nonumber \\
&+& (1-\beta)I(X_{i2},X_{R};X_{j2},Y_{k2})
\nonumber\\
c_k&=&\beta I(X_{11},X_{21};Y_{k1},\hat{Y}_R) + (1-\beta)I(X_{12},X_{22};X_R,Y_{k2})
\nonumber \\
d_k&=&\beta[I(X_{11},X_{21},\hat{Y}_{R};Y_{k1})+I(X_{11},X_{21};\hat{Y}_R)\nonumber \\
&-&I(Y_R;\hat{Y}_R)] + (1-\beta)I(X_{12},X_{22},X_{R};Y_{k2}),
\label{set-abcd}
\end{eqnarray}
}
$i,j,k\in \{1,2\}$, $i\neq j$ and $k \in \{1,2\}$, 

for all input distributions
\begin{equation}
p(x_{11})p(x_{21})p(x_{12})p(x_{22})p(x_R)p(\hat{y}_R|y_R)
\label{inputD}
\end{equation}
\end{theorem}
\begin{proof}: The detail of the proof is shown in the first part of the Appendix.
\end{proof} 
\textit{Remark 1:} The major difference between the GQF scheme and the CF scheme applied in \cite{Gunduz2010} is that relay does not perform binning after quantize its observation of the sources messages. Moreover, in GQF two destinations perform one-step joint-decoding of both messages instead of sequentially decoding the relay bin index and then the source messages.

\subsection{Achievable Rate Region Based on modified CF Scheme}
In this subsection, as a reference, the achievable rate regions based on the modified CF scheme will be shown in the HD-MARC and HD-cMACr respectively. The modification is done for two parts: First, the relay in the classic CF scheme is now half-duplex; Second, the encoding and decoding at sources and destinations are now using a single block two slots structure.

In this CF scheme, relay quantizes its observation in the end of the first slot. After that, the relay implements Wyner-Ziv binning and sends the bin index in the second slot. The destination sequentially decodes the bin index $\hat{s}\in \mathcal{S}$, quantization index $\hat{u} \in B(\hat{s})$ with the side information and finally the source messages $\hat{w}_1\in \mathcal{W}_1$ and $\hat{w}_2\in \mathcal{W}_2$ jointly from both slots reception. 

The achievable rate regions for discrete memoryless HD-MARC and HD-cMACr based on the modified CF can be summarized in the following:  
\begin{theorem}\label{th-CFD-MARC}
The following rate regions are achievable over discrete memoryless half-duplex MARC based on the modified CF scheme:
{\setlength\arraycolsep{0.1em}
\begin{eqnarray}
R_i & < & a_{1}(i)
\label{eqn-CFD111-MARC} \\
R_1+R_2 & < & c_1
\label{eqn-CFD222-MARC}
\end{eqnarray}
}
subject to
\begin{equation}
\beta [I(Y_R;\hat{Y}_R)-I(Y_{11};\hat{Y}_R)] <(1-\beta)I(X_R;Y_{12})
\label{eqn-CFDcon-MARC}
\end{equation}
where $i \in\{1,2\}$, $a_{1}(i), c_1$ are previously defined as in (\ref{set-abcd}), for all the input distributions as in (\ref{inputD}).
\end{theorem}
\begin{proof} Let one destination, say $D_2$, does not receive any signals from the HD-cMACr. Then the above result can be obtained following the proof for theorem \ref{th-CFD}.
\end{proof}
\begin{theorem}\label{th-CFD}
The following rate regions are achievable over discrete memoryless HD-cMACr with the modified CF scheme:
{\setlength\arraycolsep{0.1em}
\begin{eqnarray}
R_i & < & min\{a_{1}(i),a_{2}(i)\}
\label{eqn-CFD111} \\
R_1+R_2 & < & min \{c_1,c_2\}
\label{eqn-CFD222}
\end{eqnarray}
}
subject to
\begin{equation}
\begin{split}
max \{\beta [I(Y_R;\hat{Y}_R)-I(Y_{11};\hat{Y}_R)],  \beta[I(Y_R;\hat{Y}_R)-I(Y_{21};\hat{Y}_R)]\}
\\
< min   \{(1-\beta)I(X_R;Y_{12}),
(1-\beta)I(X_R;Y_{22})\}
\end{split}
\label{eqn-CFDcon}
\end{equation}
where $i\in\{1,2\}$, $a_{1}(i), a_{2}(i), c_1$, $c_2$ are previously defined as in (\ref{set-abcd}), for all the input distributions as in (\ref{inputD}).
\end{theorem}
\begin{proof} The outline of the proof can be found in the second part of the Appendix.
\end{proof}
\textit{Remark 2:} The GQF and the modified CF schemes provide different achievable rate regions. Note that results based on the modified CF should have (\ref{eqn-CFDcon-MARC}) and (\ref{eqn-CFDcon}) hold, which means the relay-destination link good enough to support the compression at relay to be recovered at destination(s). If this condition holds, then achievable rates based on the modified CF scheme are no less than those based on the GQF scheme. In other words, the GQF scheme cannot provide a better result in terms of achievable rate region. Therefore, when a higher sum rate is desired in the HD-MARC and HD-cMACr, (\ref{eqn-CFDcon-MARC}) and (\ref{eqn-CFDcon}) hold and a simplified relay is not required, the CF based scheme is suggested. On the other hand, if a low-cost simplified relay is preferred or (\ref{eqn-CFDcon-MARC}) and (\ref{eqn-CFDcon}) do not hold, then the GQF scheme is a superior choice.

\subsection{Special Case of The Achievable Rates Result}

In this subsection, we show that the achievable rate region for the three-node HDRC \cite{Yao2013} can be induced from the aforementioned achievable rate region for the five-node HD-cMACr.
Specifically, by taking $R_2=0$ and $X_{21}=X_{22}=Y_{21}=Y_{22}=\phi$ in HD-cMACr, a reduced three-node HDRC which contains $S_1$, $R$ and $D_1$ is considered.
\subsubsection{special case of GQF scheme}
For the GQF scheme, since $Y_{21}=Y_{22}=\phi$ and $R_2=0$, the individual rate (\ref{eqn-QFD111})
become 
{\setlength\arraycolsep{0.1em}
\begin{eqnarray}
R_1 & < & 
min \{\beta I(X_{11};Y_{11},\hat{Y}_R) + (1-\beta)I(X_{12};Y_{12}|X_R),
\nonumber \\
& & \qquad \beta[I(X_{11};Y_{11})-I(Y_R;\hat{Y}_R | X_{11},Y_{11})]
\nonumber \\
& & \qquad + (1-\beta)I(X_{12},X_R;Y_{12})\}
\label{eqn-conQFD1}
\end{eqnarray}
}
where (\ref{eqn-conQFD1}) is based on the Markov chain $(X_{11},Y_{11})\rightarrow Y_R\rightarrow \hat{Y}_R$ that $H(\hat{Y}_R|Y_R)=H(\hat{Y}_R|Y_R,X_{11},Y_{11})$.
Also the sum rate (\ref{eqn-QFD222}) can also be rewritten as
{\setlength\arraycolsep{0.1em}
\begin{eqnarray}
R_1 & < & 
min \{\beta I(X_{11};Y_{11},\hat{Y}_R) + (1-\beta)I(X_{12};Y_{12}|X_R),
\nonumber \\
& & \qquad \beta[I(X_{11},\hat{Y}_R;Y_{11}) + I(X_{11};\hat{Y}_R)-I(Y_R;\hat{Y}_R)]
\nonumber \\
& & \qquad + (1-\beta)I(X_{12},X_R;Y_{12})\}
\label{eqn-conQFD2}.
\end{eqnarray}
}
Observing that (\ref{eqn-conQFD2}) is the same as (\ref{eqn-conQFD1}). By changing the variable names accordingly, $R_1$ from the individual rate and the sum rate become the same as in theorem 1 of \cite{Yao2013}. Therefore the achievable rates based on QF scheme of \cite{Yao2013} can be treated as a special case of \textit{Theorem \ref{th-QFD}} of this work.

\subsubsection{special case of modified CF scheme}
Since in the three-node HDRC $Y_{21}=Y_{22}=\phi$,
the individual rate $R_1$ from (\ref{eqn-CFD111}) can be rewritten as:
{\setlength\arraycolsep{0.1em}
\begin{eqnarray}
R_1 & < & 
\beta I(X_{11};Y_{11},\hat{Y}_R) + (1-\beta)I(X_{12};Y_{12}|X_R)
\label{eqn-conCFD3}
\end{eqnarray}
}
where 
(\ref{eqn-conCFD3}) is from mutual information identity and $X_{21}=X_{22}=\phi$. Similarly using $R_2=0$, the sum rate inequality (\ref{eqn-CFD222}) can be rewritten as the same as (\ref{eqn-conCFD3}). Also note that the condition for the achievable rate region (\ref{eqn-CFDcon}) become
\begin{equation}
(1-\beta)I(X_R;Y_{12})>\beta I(Y_R;\hat{Y}_R)-\beta I(Y_{11};\hat{Y}_R).
\end{equation}
By changing the variable names respectively, the achievable rate region based on CF scheme of \cite{Yao2013} can also be considered as a special case of the result of \textit{Theorem \ref{th-CFD}} based on the modified CF scheme of this work.

\section{Gaussian Channels and Numerical Examples}
In this section, we extend the proposed GQF scheme and the modified CF scheme to the Gaussian Channels. Some numerical results are shown to compare the performance of the two schemes in the Half-Duplex Gaussian MARC. 

Notice that in a cMACr both destinations need to decode both messages from the sources. However, in an interference relay channel \cite{Tian2011}, each of the destination may not be interested in decoding the message from the interfered source. The sum achievable rates in a cMACr are always a minimum function of two terms that obtained from the achievability of each destination. Therefore, for clarity of the presentation and simplicity of exposition, the extended results to the half duplex Gaussian cMACr will not be shown as they have the similar performance and effect in comparing GQF and CF. 

Consider a Gaussian HD-MARC as shown in Fig.\ref{fig:MARC-phases}. Following the similar notation of \cite{Yao2013}, the channel transition probabilities specified in the below relationships:
{\setlength\arraycolsep{0.1em}
\begin{eqnarray}
y_{11}^n & = & h_{11}x_{11}^n+h_{21}x_{21}^n+z_{11}^n
\nonumber \\
y_{R}^n & = & h_{1R}x_{11}^n+h_{2R}x_{21}^n+z_{R}^n
\nonumber \\
y_{12}^m & = & h_{11}x_{12}^m+h_{21}x_{22}^m+h_{R1}x_R^m+z_{12}^m
\nonumber
\end{eqnarray}
}
where $h_i$ for $i\in\{11,21,1R,2R,R1\}$ are real constants representing the channel gain, and the channel noises $z_{11}^n,z_{R}^n$ and $z_{12}^m$ are generated independently and identically according to Gaussian distributions with zero means and unit variances. They are independent of other random variables in the model. 

The transmitters at the sources and the relay have power constraints over the transmitted sequences in each slot as the following:
\begin{eqnarray}
\frac{1}{n}\sum_{i=1}^{n}|x_{j,i}| & \leq & P_j, \text{for}\:j\in\{11,21\};
\\
\frac{1}{m}\sum_{i=1}^{m}|x_{k,i}| & \leq & P_k,\text{for}\:k\in\{12,22,R\},
\end{eqnarray}
where $|x|$ shows the absolute value of $x$. 

In this work, it is assumed that in the Gaussian channels all the codebooks used are generated according to some zero-mean Gaussian distributions. Notice that these input distributions are not necessarily the optimal distributions which maximize the achievable rate. Nevertheless, the Gaussian codebooks are still used since they are the most widely assumption in the literature and make the analysis of characterizing the achievable rates tractable for illustration purpose.

Let $X_i$ for $i\in\{11,21,12,22,R\}$, $Z_j$ for $j\in\{11,12,R\}$ and $Z_Q$ be generic random variables which are Gaussian with zero mean and are mutually independent. The variances of $X_i$, $Z_j$ and $Z_Q$ are $P_i$, 1 and $\sigma_Q^2$ respectively. The random variable $Y_k$ denotes the channel output where $k\in\{11,12,R\}$.
$\hat{Y}_R$ is the estimation of $Y_R$. The following equations show the relationships between the introduced random variables:
{\setlength\arraycolsep{0.1em}
\begin{eqnarray}
Y_{11} & = & h_{11}X_{11}+h_{21}X_{21}+Z_{11};
\\
Y_{12} & = & h_{11}X_{12}+h_{21}X_{22}+h_{R1}X_R+Z_{12};
\\
Y_{R} & = & h_{1R}X_{11}+h_{2R}X_{21}+Z_{R};
\\
\hat{Y}_R &=& Y_R+Z_Q. \label{eqn-RtoRhat}
\end{eqnarray}
}
In the following, the achievable rate region for the Gaussian setup with the GQF scheme is characterized.
\begin{proposition}
The following rates are achievable for the Gaussian HD-MARC by using the GQF scheme:
{\setlength\arraycolsep{0.1em}
\begin{eqnarray}
R_i &<&  \underset{\sigma_Q^2,\beta}{max} \;\; min \{\frac{\beta}{2}log(1+h_{i1}^2P_{i1}+\frac{h_{iR}^2P_{i1}}{1+\sigma_Q^2})
\nonumber \\
&&+\frac{1-\beta}{2}log(1+h_{i1}^2P_{i2}),
\nonumber \\
&&\frac{\beta}{2}log(\frac{(1+h_{i1}^2P_{i1})\sigma_Q^2}{1+\sigma_Q^2})
\nonumber \\
&&+\frac{1-\beta}{2}log(1+h_{i1}^2P_{i2}+h_{R1}^2P_R)\}
\label{eqn-GQF-GauMARCindi}\\
R_1 &+& R_2 <   \underset{\sigma_Q^2,\beta}{max} \;\; min \{\frac{\beta}{2}log(1+h_{11}^2P_{11}+h_{21}^2P_{21}
\nonumber \\
&&+\frac{(h_{11}h_{2R}-h_{1R}h_{21})^2P_{11}P_{21}+h_{1R}^2P_{11}+h_{2R}^2P_{21}}{1+\sigma_Q^2})
\nonumber \\ 
&&+\frac{1-\beta}{2}log(1+h_{11}^2P_{12}+h_{21}^2P_{22}),
\nonumber \\
&&\frac{\beta}{2}log(\frac{(1+h_{11}^2P_{11}+h_{21}^2P_{21})\sigma_Q^2}{1+\sigma_Q^2})
\nonumber \\
&&+\frac{1-\beta}{2}log(1+h_{11}^2P_{12}+h_{21}^2P_{22}+h_{R1}^2P_R)\}
\label{eqn-GQF-GauMARCsum}
\end{eqnarray}
}
where i = 1,2 and $\sigma_Q^2$ is the relay quantization factor.
\end{proposition}
\emph{Remark 3}:\; Within the achievable sum rate (\ref{eqn-GQF-GauMARCsum}), the two min terms can be treated as two functions of $\sigma_Q^2$. Similarly as \cite{Yao2013}, denote the first term as $I_1(\sigma_Q^2)$ and the second term as $I_2(\sigma_Q^2)$. The effect of the different values of the $\sigma_Q^2$ on the achievable sum rate is shown in the Fig.\ref{fig:Sigma-q-GQF}. It can be seen that, for fixed $\beta$, $I_1(\sigma_Q^2)$ is a simple decreasing function and  $I_2(\sigma_Q^2)$ is an increasing function. In other words, the first order derivative of  $I_1(\sigma_Q^2)$ is always negative and that of  $I_2(\sigma_Q^2)$ is always positive. Let $I_1(\sigma_Q^2)=I_2(\sigma_Q^2)$, the value of $\sigma_Q^2$ that maximizes the sum rate can be obtained.

\begin{figure}[t]
\centering
\includegraphics[scale=0.225]{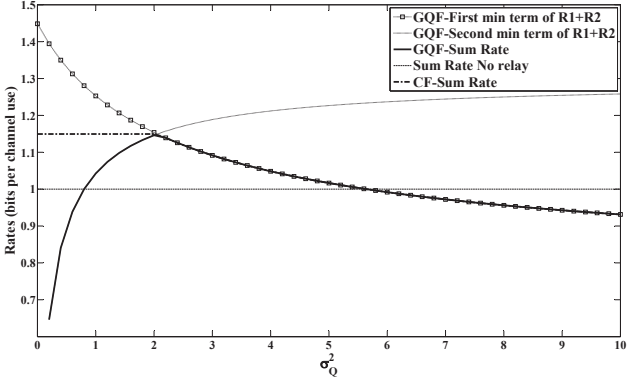}
\caption{Achievable Rates of GQF and CF based scheme with variant $\sigma_q^2,P_{11}=P_{12}=1,P_{21}=P_{22}=1,P_R=1$, and if no relay source powers $P_1=P_2=1.5$, channel gain are $h_{11}=h_{21}=1, h_{1R}=3, h_{2R}=0.5, h_{R1}=3,\beta=0.5$}
\label{fig:Sigma-q-GQF}
\end{figure}

Next, the achievable rate region for the Gaussian setup with the modified CF scheme will be described for the HD-MARC. As stated before, a Gaussian quantization codebook is assumed for illustration purpose and the optimality is not claimed here.
\begin{proposition}
The following rates are achievable for the Gaussian HD-MARC by using the modified CF scheme:
{\setlength\arraycolsep{0.1em}
\begin{eqnarray}
R_i&<&\frac{\beta}{2}log(1+h_{i1}^2P_{i1}+\frac{h_{iR}^2P_{i1}}{1+\sigma_Q^2})
\nonumber \\
&&+\frac{1-\beta}{2}log(1+h_{i1}^2P_{i2}),
\label{eqn-CF-GauMARCindi} 
\\
R_1 &+& R_2 < \frac{\beta}{2}log(1+h_{11}^2P_{11}+h_{21}^2P_{21}
\nonumber \\
&+&\frac{(h_{11}h_{2R}-h_{1R}h_{21})^2P_{11}P_{21}+h_{1R}^2P_{11}+h_{2R}^2P_{21}}{1+\sigma_Q^2})
\nonumber \\
&+&\frac{1-\beta}{2}log(1+h_{11}^2P_{12}+h_{21}^2P_{22})
\label{eqn-CF-GauMARCsum}
\end{eqnarray}
}
where i=1,2 and 
\begin{eqnarray}
\sigma_Q^2 >
\frac{1+\frac{h_{1R}^2P_{11}+h_{2R}^2P_{21}+(h_{11}h_{2R}-h_{1R}h_{21})^2P_{11}P_{21}}{1+h_{11}^2P_{11}+h_{21}^2P_{21}}}{(1+\frac{h_{R1}^2P_R}{1+h_{11}^2P_{12}+h_{21}^2P_{22}})^{\frac{1-\beta}{\beta}}-1}
\label{eqn-CF-GauMARC}
\end{eqnarray}
\end{proposition}
\emph{Remark 4}:\; The constraint condition (\ref{eqn-CFDcon-MARC}) in the discrete memoryless channel guarantees the quantized observation at relay can be recovered at destination. Here in the Gaussian setup, it is translated to the condition of (\ref{eqn-CF-GauMARC}). It can be seen that a minimum value of $\sigma_Q^2$ is required for the modified CF scheme. This is essentially due to the characteristic of $Z_Q$ in (\ref{eqn-RtoRhat}) where a larger value of $\sigma_Q^2$ will result $\hat{Y}_R$ to be a more degraded version of $Y_R$ or in other words a more compressed signal at relay.

The achievable sum rate term (\ref{eqn-CF-GauMARCsum}) based on the CF scheme is the same as the first min term of (\ref{eqn-GQF-GauMARCsum}) based on GQF scheme when the constraint on $\sigma_Q^2$ in (\ref{eqn-CF-GauMARC}) is satisfied. Fig.\ref{fig:Sigma-q-GQF} also shows the sum rates based on modified CF scheme in the HD-MARC based on different $\sigma_Q^2$. Notice that if relay uses a good quantizer or relay has a good estimate $\hat{Y}_R$ of $Y_R$ such that $\sigma_Q^2$ is less than the right-hand side of (\ref{eqn-CF-GauMARC}). Then, a higher sum rate cannot be achieved. This is due to the channel between relay and destination is limiting the compressed observation at relay to be recovered at destination. In other words, $\hat{Y}_R$ has a higher rate than the channel between relay and destination can support. Thus for smaller value of $\sigma_Q^2$, the sum rate is the same as the one taken from the constraint condition.

\begin{figure}[t]
\centering
\includegraphics[scale=0.23]{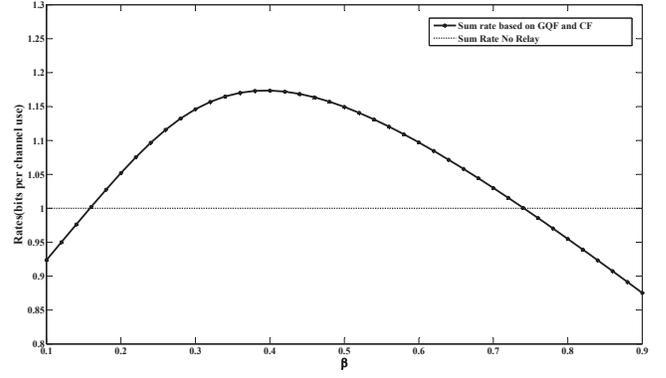}
\caption{Achievable Rates of GQF and CF based scheme with variant $\beta$}
\label{fig:beta-factor}
\end{figure}

As defined previously, $\beta$ is the ratio of the two slots taken in each block. The impact of the factor $\beta$ on the achievable rates that based on the GQF scheme in the HD-MARC channel is shown in Fig. \ref{fig:beta-factor} where we assume the same power and channel gain as Fig. \ref{fig:Sigma-q-GQF}. It can be seen that under such a channel state in order to maximize the achievable sum rate the length of each slot should be carefully chosen. Notice that if the relay quantization random variable $\sigma_Q^2$ was chosen to satisfy the constraint (\ref{eqn-CF-GauMARC}), and then the achievable sum rates based on CF is the same as those based on GQF. As also shown in the Fig. \ref{fig:beta-factor}, both GQF and CF schemes outperform the case where no relay is available in the channel.

Comparing the achievable sum rates of (\ref{eqn-GQF-GauMARCsum}) with optimized $\sigma_Q^2$ and (\ref{eqn-CF-GauMARCsum}) for the HD-MARC based on both GQF and CF under the Gaussian setup, the sum rates are same. In general, without optimizing the relay quantization factor $\sigma_Q^2$, the CF scheme outperforms the GQF scheme. However, by choosing the optimized value of $\sigma_Q^2$, the GQF scheme, in which a low-cost simplified relay is used, is able to provide similar sum rate performance as the more complicated but sophisticated CF scheme.

\section{Conclusion}

In this paper the Half-Duplex relaying in the Multiple Access Relay Channel and the compound Multiple Access Channel with a relay has been studied. A variation of the QF scheme, the GQF scheme, based on single block coding has been proposed. The GQF scheme employs joint decoding at destinations and uses a low-cost relay where it only quantizes the received signal after first slot and forwards it in the second slot without Wyner-Ziv binning. For comparison purpose, a modified CF scheme was also introduced. The achievable rate regions were obtained based on GQF scheme and CF scheme for HD-MARC and HD-cMACr, respectively. It is also shown that the achievable rate regions for the three-node HDRC can be treated as specical cases of our results obtained for the five-node channel. As a further development, the achievable rate results from discrete memoryless channels were also extended to the Half-Duplex Gaussian MARC. Some numerical examples were provided for the purpose of the performance comparison. The results indicate that the proposed GQF scheme can provide a similar performance as the CF scheme with only a simplified low-cost relay.

\section*{Appendix}
\subsection{Proof of Theorem \ref{th-QFD}}
Assume the source messages $W_{1}$ and $W_{2}$ are independent of each other. Each message $W_{i}$, $i=1,2$, is uniformly distributed in its message set $\mathcal{W}_i = [1 : 2^{lR_i}]$.

\subsubsection{Codebook Generation}

Assume the joint pmf factors as
\begin{equation}
\begin{split}
&p(x_{11})p(x_{21})p(x_{12})p(x_{22})p(x_R)p(\hat{y}_R|y_R)\\
&p(y_{11},y_{21},y_R|x_{11},x_{12})p(y_{21},y_{22}|x_{12},x_{22},x_R).
\end{split}
\end{equation}
Fix any input distribution
$$p(x_{11})p(x_{21})p(x_{12})p(x_{22})p(x_R)p(\hat{y}_R|y_R).$$
Randomly and independently generate
\begin{itemize}

    \item $2^{lR_1}$ codewords $x_{11}^{n}(w_1)$, $w_1\in\mathcal{W}_1$, each according to $\prod _{i=1}^{n} p_{X_{11}} (x_{11,i}(w_1))$;
    \item $2^{lR_2}$ codewords $x_{21}^{n}(w_2)$, $w_2\in\mathcal{W}_2$, each according to $\prod _{i=1}^{n} p_{X_{21}} (x_{21,i}(w_2))$;
    \item $2^{lR_1}$ codewords $x_{12}^{m}(w_1)$, $w_1\in\mathcal{W}_1$, each according to $\prod _{i=1}^{m} p_{X_{12}} (x_{12,i}(w_1))$;
    \item $2^{lR_2}$ codewords $x_{22}^{m}(w_2)$, $w_2\in\mathcal{W}_2$, each according to $\prod _{i=1}^{m} p_{X_{22}} (x_{22,i}(w_2))$;
    \item $2^{lR_U}$ codewords $x_{R}^{m}(u)$, $u\in\mathcal{U}=\{1,2,\dots 2^{lR_U}\}$, each according to $\prod _{i=1}^{m} p_{X_{R}} (x_{R,i}(u))$.
\end{itemize}
Calculate the marginal distribution
$$p(\hat{y}_R)=\sum_{x_{11}\in \mathcal{X_{11}} ,x_{21}\in \mathcal{X_{21}},y_{11}\in \mathcal{Y_{21}},y_{21}\in \mathcal{Y_{21}},y_{R}\in \mathcal{Y_R}} p(\hat{y}_R|y_R)$$
$$\qquad p(y_R,y_{11},y_{21}|x_{11},x_{21})p(x_{11})p(x_{21}).$$
Randomly and independently generate $2^{lR_U}$ codewords $\hat{y}_{R}^{n}(u)$, each according to $\prod _{i=1}^{n} p_{\hat{Y}_{R}} (\hat{y}_{R,i}(u))$.

\subsubsection{Encoding}
To send messages $w_i$, the source node $S_i$ transmits $x_{i1}^{n}(w_i)$ in the first slot and $x_{i2}^{m}(w_i)$ in the second slot, where $i=1,2$. Let $\epsilon' \in (0,1)$ . After receiving $y_R^n$ at the end of the first slot, the relay tries to find a unique $u\in\mathcal{U}$ such that 
\begin{equation}
(y_R^n,\hat{y}_R^n(u))\in \mathcal{T}_{\epsilon'}^n(Y_R,\hat{Y}_R)
\end{equation}
where $\mathcal{T}_{\epsilon}^n(Y_R,\hat{Y}_R)$ is the $\epsilon$-strongly typical set as defined in \cite{Lim2011}. If there are more than one such $u$, randomly choose one in $\mathcal{U}$. The relay then sends $x_R^m(u)$ in the second slot.

\subsubsection{Decoding}

Destinations start decoding the messages after the second slot finishes. Let $\epsilon'<\epsilon<1$. Upon receiving in both slots, $D_1$ and $D_2$ tries to find a unique pair of the messages $\hat{w}_1\in\mathcal{W}_1$ and $\hat{w}_2\in\mathcal{W}_2$ such that 
\begin{eqnarray}
      (x_{11}^n(\hat{w}_1),x_{21}^n(\hat{w}_2),y_{11}^n,\hat{y}_R^n(u)) \in \mathcal{T}_\epsilon^n(X_{11},X_{21},Y_{11},\hat{Y}_R)\\
      (x_{12}^m(\hat{w}_1),x_{22}^m(\hat{w}_2),x_R^m(u),y_{12}^m) \in \mathcal{T}_\epsilon^m(X_{12},X_{22},X_R,Y_{12})
\end{eqnarray}
and
\begin{eqnarray}
      (x_{11}^n(\hat{w}_1),x_{21}^n(\hat{w}_2),y_{21}^n,\hat{y}_R^n(u)) \in \mathcal{T}_\epsilon^n(X_{11},X_{21},Y_{21},\hat{Y}_R)\\
      (x_{12}^m(\hat{w}_1),x_{22}^m(\hat{w}_2),x_R^m(u),y_{22}^m) \in \mathcal{T}_\epsilon^m(X_{12},X_{22},X_R,Y_{22})
\end{eqnarray}

for some $u\in\mathcal{U}$.

\subsubsection{Probability of Error Analysis}

Let $W_i$ denote the message sent from source node $S_i, i=1,2$. $U$ represents the index chosen by the relay $R$. The probability of error averaged over $W_1$,$W_2$, $U$ over all possible codebooks is defined as
{\setlength\arraycolsep{0.1em}
\begin{align}
Pr(\mathcal\epsilon) &  
= Pr(\hat{W}_1\neq 1 \cup \hat{W}_2\neq 1 | W_1=1, W_2=1)\label{poe}.   
\end{align}
}
The (\ref{poe}) is based on the symmetry of the codebook construction and the fact that the messages $W_1$ and $W_2$ are chosen uniformly from $\mathcal{W}_1$ and $\mathcal{W}_2$, the overall probability of error is equal to the probability of error when $W_1=1$ and $W_2=1$ were selected as the message indices.
Define three events $\mathcal{E}_{0}$ $\mathcal{E}_{2,(w_1,w_2)}$ and $\mathcal{E}_{2,(w_1,w_2)}$ which are described in the following with $i=1,2$:
{\setlength\arraycolsep{0.1em}
\begin{align}
\mathcal{E}_{0}  &:=  \{((Y_R^n,\hat{Y}_R^n(u))\notin  \mathcal{T}_{\epsilon'}^n(Y_R\hat{Y}_R)), \text {for all} \: u \} 
\\     
\mathcal{E}_{i,(w_1,w_2)}    &:=
\{ (X_{11}^n(w_1),X_{21}^n(w_2),Y_{i1}^n,\hat{Y}_R^n(u)) 
\nonumber \\
& \qquad \in \mathcal{T}_{\epsilon}^n(X_{11}X_{21}Y_{i1}\hat{Y}_R) \:\: \text{and}
\nonumber \\
&\qquad (X_{12}^m(w_1),X_{22}^m(w_2),X_R^m(u),Y_{i2}^m) 
\nonumber \\
&\qquad \in \mathcal{T}_{\epsilon}^m(X_{11}X_{21}X_RY_{i2}) \; \text{for some}\: u \}.
\end{align}
}
Then the probability of error can be rewritten as
{\setlength\arraycolsep{0.1em}
\begin{eqnarray}
Pr(\mathcal\epsilon) 
& \leq & Pr (\mathcal{E}_{0}|W_1=1,W_2=1) 
\nonumber \\
& & + Pr( (\mathcal{E}_{1,(1,1)} \cap \mathcal{E}_{2,(1,1)})^c\cap\mathcal{E}_{0}^c|W_1=1,W_2=1)
\nonumber \\
&  & + Pr(\cup_{(w_1,w_2)\in\mathcal{A}} \mathcal{E}_{1,(w_1,w_2)}|W_1=1,W_2=1)
\nonumber \\
&  & + Pr(\cup_{(w_1,w_2)\in\mathcal{A}} \mathcal{E}_{2,(w_1,w_2)}|W_1=1,W_2=1),
\label{ineqn-err}
\end{eqnarray}
}
where $\mathcal{A}:=\{(w_1,w_2)\in\mathcal{W}_1\times \mathcal{W}_2:(w_1,w_2)\neq (1,1)\}$. Assume $\beta$ is fixed, then by covering lemma \cite{Gamal2010}, $Pr(\mathcal{E}_{0}|W_1=1,W_2=1)\rightarrow 0$ when $l\rightarrow \infty$, if 
\begin{equation}
    R_U > \beta I(Y_R,\hat{Y}_R) + \delta(\epsilon') \label{eq-rfindu}
\end{equation}
where $\delta(\epsilon')\rightarrow 0$ as $\epsilon'\rightarrow 0$. By the conditional typicality lemma \cite{Gamal2010}, $Pr( (\mathcal{E}_{1,(1,1)} \cap \mathcal{E}_{2,(1,1)})^c\cap\mathcal{E}_{0}^c|W_1=1,W_2=1) \rightarrow 0$ as $l\rightarrow \infty$. 
Due to the space limitation and the similar fashion can be used to analyze the probability of error for both the third line term and the fourth line term of (\ref{ineqn-err}), only the analysis for the third line term $Pr(\cup_{(w_1,w_2)\in\mathcal{A}} \mathcal{E}_{1,(w_1,w_2)}|W_1=1,W_2=1)$  will be shown in this proof.
In addition, some standard probability error analysis is also omitted here and only those important steps were kept in the following.

For fixed $\beta = \frac{n}{l}$,  $1-\beta = \frac{m}{l}$, if $l\rightarrow\infty$, $\epsilon\rightarrow 0$ and the following inequalities hold:
{\setlength\arraycolsep{0.1em}
\begin{eqnarray}
R_1 & < & \beta I(X_{11};X_{21},Y_{11},\hat{Y}_{R}) 
\nonumber \\
& & + (1-\beta)I(X_{12};X_{22},X_{R},Y_{12})
\\
R_1 + R_U & < & \beta[I(X_{11},\hat{Y}_{R};X_{21},Y_{11})+I(X_{11};\hat{Y}_R)] 
\nonumber \\
& & + (1-\beta)I(X_{12},X_{R};X_{22},Y_{12})
\\
R_2 & < & \beta I(X_{21};X_{11},Y_{11},\hat{Y}_{R}) 
\nonumber \\
& & +  (1-\beta)I(X_{22};X_{12},X_{R},Y_{12})
\\
R_2 + R_U & < & \beta[I(X_{21},\hat{Y}_{R};X_{11},Y_{11})+I(X_{21};\hat{Y}_R)] 
\nonumber \\
& & + (1-\beta)I(X_{22},X_{R};X_{12},Y_{12})
\\
R_1 + R_2 & < & \beta [I(X_{11},X_{21};Y_{11},\hat{Y}_{R})+I(X_{11};X_{21})] 
\nonumber \\
& & + (1-\beta)I(X_{12},X_{22};X_{R},Y_{12})
\\
R_1 + R_2 + R_U & < & \beta [I(X_{11},X_{21},\hat{Y}_{R};Y_{11})+I(X_{11},X_{21};\hat{Y}_{R})] 
\nonumber \\
&  & + (1-\beta)[I(X_{12},X_{22},X_{R};Y_{12})
\nonumber \\
&  & \qquad + I(X_{12},X_{22};X_{R})],
\end{eqnarray}
}
then $Pr(\cup_{(w_1,w_2)\in\mathcal{A}} \mathcal{E}_{1,(w_1,w_2)}|W_1=1,W_2=1)\rightarrow 0$. Note that since the messages and the codebook have been independently generated, the above inequalities can be further simplified by substituting $I(X_{22};X_R)=0$, $I(X_{12};X_R)=0$, $I(X_{11};X_{21})=0$, $I(X_{12};X_{22})=0$ and $I(X_{12},X_{22};X_R)=0$
In the last, by taking out $R_U$ according to (\ref{eq-rfindu}), the following inequalities define the achievable rates corresponding to $D_1$:
{\setlength\arraycolsep{0.1em}
\begin{eqnarray}
R_i & < &  min \{ \beta I(X_{i1};X_{j1},Y_{11},\hat{Y}_{R}) 
\nonumber \\
& & \qquad + (1-\beta)I(X_{i2};X_{j2},X_{R},Y_{12}),
\nonumber \\
 &  & \qquad\beta[I(X_{i1},\hat{Y}_{R};X_{j1},Y_{11})+I(X_{i1};\hat{Y}_R)
\nonumber \\
& &  \qquad -I(Y_R;\hat{Y}_R)] + (1-\beta)I(X_{i2},X_{R};X_{j2},Y_{12}) \}
\nonumber
\end{eqnarray}
\begin{eqnarray} 
R_1 + R_2 & < & min \{ \beta I(X_{11},X_{21};Y_{11},\hat{Y}_{R}) 
\nonumber \\
& & \qquad + (1-\beta)I(X_{12},X_{22};X_{R},Y_{12}),
\nonumber \\
&  & \qquad\beta [I(X_{11},X_{21},\hat{Y}_{R};Y_{11})+I(X_{11},X_{21};\hat{Y}_{R})\nonumber \\
&  & \qquad -I(Y_R;\hat{Y}_R)] + (1-\beta)I(X_{12},X_{22},X_{R};Y_{12}) \}
\nonumber
\end{eqnarray}
}
where $i=1,2$ and $j=\{1,2 | \;i\neq j\}$.
Similarly the achievable rate results for $D_2$ can be obtained. Therefore, the probability of error $P(\mathcal\epsilon) \rightarrow 0$ if all the achievable inequalities corresponding to $D_1$ and $D_2$ hold. This completes the proof and those achievable inequalities are shown in Theorem 2.

\subsection{Outline of Proof for the CF based Achievable Rate Region}
Due to the space limitation and the similarity for the proof of CF and GQF based achievable rate regions, the detailed proof is omitted in this subsection. Note that there are two major differences in the CF scheme comparing to the GQF scheme: First, after $R$ quantizes the received signal from first slot with a rate $R_U$, it applies the Wyner-Ziv binning to further partition the set of $\mathcal{U}$ into $2^{lR_S}$ equal size bins and send the bin index $S$ with $X_{R}(s)$ in the second slot; Second, each destination performs step decoding for the bin index $\hat{s}$, ${u}$ and $(\hat{w_1},\hat{w_2})$ sequentially. In the last step decoding of the CF scheme, the decoder jointly decodes both messages from the signals received in both slots.

\bibliographystyle{IEEEtran}
\bibliography{reference}

\end{document}